\def\llncs{1}
\def\fullpage{0}
\def\anonymous{0}
\def\authnote{1}
\def\notxfont{0}
\def\submission{1}
\def\cameraready{1}
\def\arxiv{0}

\ifnum\llncs=1
	\documentclass[envcountsect,a4paper,runningheads,10pt]{llncs}
\else
	\documentclass[letterpaper,hmargin=1.05in,vmargin=1.05in]{article}
			\ifnum\fullpage=1
		\usepackage{fullpage}
		\fi
\fi

\usepackage[%
  colorlinks=true,
  citecolor=blue,
  pagebackref=true
]{hyperref}

\usepackage{amsmath, amsfonts, amssymb, mathtools,amscd}

\usepackage{amsthm}

\usepackage{lmodern}
\usepackage[T1]{fontenc}
\usepackage[utf8]{inputenc}

\usepackage{arydshln} 
\usepackage{url}
\usepackage{ifthen}
\usepackage{bm}
\usepackage{multirow}
\usepackage[dvips]{graphicx}
\usepackage[usenames]{color}
\usepackage{xcolor,colortbl} 
\usepackage{threeparttable}
\usepackage{comment}
\usepackage{paralist,verbatim}
\usepackage{cases}
\usepackage{booktabs}
\usepackage{braket}
\usepackage{cancel} 
\usepackage{ascmac} 
\usepackage{framed}
\usepackage{authblk}
\usepackage{pifont}
\usepackage{autonum}
\definecolor{darkblue}{rgb}{0,0,0.6}
\definecolor{darkgreen}{rgb}{0,0.5,0}
\definecolor{maroon}{rgb}{0.5,0.1,0.1}
\definecolor{dpurple}{rgb}{0.2,0,0.65}

\usepackage[capitalise,noabbrev]{cleveref}
\usepackage[absolute]{textpos}
\usepackage[final]{microtype}
\usepackage[absolute]{textpos}
\usepackage{everypage}

\newtheoremstyle{thicktheorem}%
{\topsep}
{\topsep}
{\itshape}{}%
{\bfseries}%
{.}
{ }%
{\thmname{#1}\thmnumber{ #2}%
		\thmnote{ (#3)}%
}

\newtheoremstyle{remark}
{\topsep}
{\topsep}
	{}
	{}
	{}
	{.}
	{ }
	{\textit{\thmname{#1}}\thmnumber{ #2}
			\thmnote{ (#3)}%
	}

\ifnum\llncs=0
	\theoremstyle{thicktheorem}
	\newtheorem{theorem}{Theorem}[section]
	\newtheorem{lemma}[theorem]{Lemma}

	\newtheorem{definition}[theorem]{Definition}
	\newtheorem{game}[theorem]{Game}

	\theoremstyle{remark}
	
	\newtheorem{remark}[theorem]{Remark}

\else
\fi

	\crefname{theorem}{Theorem}{Theorems}
	\crefname{assumption}{Assumption}{Assumptions}
	\crefname{construction}{Construction}{Constructions}
	\crefname{corollary}{Corollary}{Corollaries}
	\crefname{conjecture}{Conjecture}{Conjectures}
	\crefname{definition}{Definition}{Definitions}
	\crefname{exmaple}{Example}{Examples}
	\crefname{experiment}{Experiment}{Experiments}
	\crefname{counterexample}{Counterexample}{Counterexamples}
	\crefname{lemma}{Lemma}{Lemmata}
	\crefname{observation}{Observation}{Observations}
	\crefname{proposition}{Proposition}{Propositions}
	\crefname{remark}{Remark}{Remarks}
	\crefname{claim}{Claim}{Claims}
	\crefname{fact}{Fact}{Facts}
	\crefname{note}{Note}{Notes}

\ifnum\llncs=1
 \crefname{appendix}{App.}{Appendices}
 \crefname{section}{Sec.}{Sections}
\else
\fi

\ifnum\llncs=1
\renewcommand*{\backref}[1]{}
\else
	\renewcommand*{\backref}[1]{(Cited on page~#1.)}
	\ifnum\notxfont=1
	\else
		\usepackage{newtxtext}
	\fi
\fi

\ifnum\authnote=0  
\newcommand{\mor}[1]{}
\newcommand{\taiga}[1]{}
\newcommand{\ryo}[1]{}
\newcommand{\takashi}[1]{}

\else
\newcommand{\mor}[1]{$\ll$\textsf{\color{red} Tomoyuki: { #1}}$\gg$}
\newcommand{\taiga}[1]{$\ll$\textsf{\color{magenta} Taiga: { #1}}$\gg$}
\newcommand{\takashi}[1]{$\ll$\textsf{\color{orange} Takashi: { #1}}$\gg$}
\newcommand{\ryo}[1]{$\ll$\textsf{\color{darkgreen} Ryo: { #1}}$\gg$}

\fi

\newcommand{\BQP}{\mathbf{BQP}}
\newcommand{\QMA}{\mathbf{QMA}}
\newcommand{\NP}{\mathbf{NP}}

\newcommand{\StateGen}{\mathsf{StateGen}}
\newcommand{\hyb}{\mathsf{Hyb}}
\newcommand{\Prob}{\mbox{Pr}}

\newcommand{\cA}{\mathcal{A}}

\newcommand{\cC}{\mathcal{C}}

\def\makeuppercase#1{
\expandafter\newcommand\csname tl#1\endcsname{\widetilde{#1}}
}

\def\makelowercase#1{
\expandafter\newcommand\csname tl#1\endcsname{\widetilde{#1}}
}

\newcommand*{\algo}[1]{\ensuremath{\mathsf{#1}}}

\newenvironment{boxfig}[2]{\begin{figure}[#1]\fbox{\begin{minipage}{0.97\linewidth}
                        \vspace{0.2em}
                        \makebox[0.025\linewidth]{}
                        \begin{minipage}{0.95\linewidth}
            {{
                        #2 }}
                        \end{minipage}
                        \vspace{0.2em}
                        \end{minipage}}}{\end{figure}}

\newcommand{\bit}{\{0,1\}}

\newcommand{\Gen}{\algo{Gen}}

\newcommand{\Sign}{\algo{Sign}}

\newcommand{\negl}{{\mathsf{negl}}}

\usepackage{fancyhdr}

\ifnum\llncs=1
\let\oldvec\vec
\let\vec\oldvec
\makeatletter
\renewcommand*\l@author[2]{}
\renewcommand*\l@title[2]{}
\makeatletter

\else
\fi    

\theoremstyle{remark}

\title{
\textbf{Quantum Commitments and Signatures without One-Way Functions}
}

\begin{document}

\ifnum\anonymous=1
\author{\empty}\institute{\empty}
\else
%
%
\ifnum\llncs=1
\author{
	Tomoyuki Morimae\inst{1} \and Takashi Yamakawa\inst{1,2}
}
\institute{
	Yukawa Institute for Theoretical Physics, Kyoto University, Japan \and NTT Corporation, Japan
}
\else
%
%
\ifnum\noclassic=1
\author[1]{Ryo Nishimaki}
\author[1]{Takashi Yamakawa}
\affil[1]{{\small NTT Secure Platform Laboratories, Tokyo, Japan}\authorcr{\small \{ryo.nishimaki.zk,takashi.yamakawa.ga\}@hco.ntt.co.jp}}
\else
\author[1]{Taiga Hiroka}
\author[1,2]{\hskip 1em Tomoyuki Morimae}
\author[3]{\hskip 1em Ryo Nishimaki}
\author[3]{\hskip 1em Takashi Yamakawa}
\affil[1]{{\small Yukawa Institute for Theoretical Physics, Kyoto University, Japan}\authorcr{\small \{taiga.hiroka,tomoyuki.morimae\}@yukawa.kyoto-u.ac.jp}}
\affil[2]{{\small PRESTO, JST, Japan}}
\affil[3]{{\small NTT Secure Platform Laboratories, Tokyo, Japan}\authorcr{\small \{ryo.nishimaki.zk,takashi.yamakawa.ga\}@hco.ntt.co.jp}}
\fi
\renewcommand\Authands{, }
\fi 
\fi

\ifnum\llncs=1
\date{}
\else
\date{\today}
\fi
\maketitle


\begin{abstract}
In the classical world, the existence of commitments is equivalent to the existence of one-way functions.
In the quantum setting, on the other hand, commitments are not known to imply one-way functions,
but all known constructions of quantum commitments use at least one-way functions.
Are one-way functions really necessary for commitments in the quantum world?
In this work, we show that non-interactive quantum commitments (for classical messages) with
computational hiding and statistical binding  
exist if pseudorandom quantum states exist.
Pseudorandom quantum states are sets of quantum states that are efficiently generated but their polynomially many copies are computationally
indistinguishable from the same number of copies of Haar random states [Ji, Liu, and Song, CRYPTO 2018].
It is known that pseudorandom quantum states exist even if $\BQP=\QMA$
(relative to a quantum oracle) [Kretschmer, TQC 2021],
which means that
pseudorandom quantum states can exist even if
no quantum-secure classical cryptographic primitive exists.
Our result therefore shows that quantum commitments can exist
even if no quantum-secure classical cryptographic primitive exists.
In particular, quantum commitments can exist even if
no quantum-secure one-way function exists.
In this work, we also consider digital signatures, which are other fundamental primitives in cryptography.
We show that one-time secure digital signatures with quantum public keys exist if pseudorandom quantum states exist.
In the classical setting, the existence of digital signatures is equivalent to the existence of one-way functions.
Our result, on the other hand, shows that quantum signatures can exist even if no quantum-secure
classical cryptographic primitive (including quantum-secure one-way functions) exists.
\end{abstract}

\ifnum\cameraready=1
\else
\ifnum\llncs=1
\else
\newpage
  \setcounter{tocdepth}{2}      
  \setcounter{secnumdepth}{2}   
  \setcounter{page}{0}          
  \tableofcontents
  \thispagestyle{empty}
  \clearpage
\fi
\fi

\section{Introduction}
\subsection{Background}
Commitments~\cite{C:Blum81} are one of the most central primitives in cryptography.
Assume that a sender wants to commit a message $m$ to a receiver. The sender
encrypts it and sends it to the receiver.
Later, the sender sends a key so that the receiver can open the message $m$.
Before the sender sends the key, the receiver should not be able to learn the message $m$, which is called
{\it hiding}.
Furthermore, the sender should not be able to change the message later once the sender commits it, which is called
{\it binding}.
(Imagine that the sender's message is put in a safe box and sent to the receiver. 
The receiver cannot open it until the receiver receives the key, and the sender cannot change the
message in the safe box once it is sent to the receiver.)
In cryptography, there are two types of definitions for security. One is statistical security
and the other is computational security. 
Statistical security means that 
it is secure against any computationally-unbounded adversary, while computational security
means that it is secure against adversaries restricted to polynomial-time classical/quantum computations.
It is easy to see that both hiding and binding cannot be statistical
at the same time in the classical setting,~\footnote{
If a commitment scheme is statistically binding, there exists at most one message to which a commitment can be opened except for a negligible probability. This unique message can be found by a brute-force search, which means that the scheme is not statistically hiding. 
} and therefore one of them
has to be based on a computational assumption. In other words, in a computationally hiding commitment scheme,
a malicious receiver can learn the message $m$ before the opening
if its computational power is unbounded, and
in a computationally binding commitment scheme, a malicious sender can change its committed message later if its
computational power is unbounded.
For the computational assumption, the existence of one-way functions
is known to be equivalent to the existence of 
classical commitments~\cite{Nao91,HILL99}.
The existence of one-way functions 
is considered the weakest assumption in classical cryptography, because virtually all complexity-based classical cryptographic primitives are known to imply the existence of one-way functions~\cite{STOC:LubRac86,FOCS:ImpLub89,STOC:ImpLevLub89}. 

The history of quantum information has demonstrated that 
utilizing quantum physics in information processing achieves many advantages.
In particular, it has been shown in quantum cryptography that
quantum physics can weaken cryptographic assumptions.
For example, if quantum states are transmitted, statistically-secure key distribution is possible~\cite{BB84}, 
although it is impossible classically.
Furthermore, oblivious transfer
is possible with only (quantum-secure) one-way functions when quantum states 
are transmitted\break~\cite{C:BCKM21b,EC:GLSV21,FOCS:CreKil88,C:BBCS91,MS94,STOC:Yao95,C:DFLSS09}.
Classically, it is known to be impossible to construct oblivious transfer from only one-way 
functions~\cite{C:ImpRud88}.\footnote{
\cite{C:ImpRud88} showed the impossibility of \emph{relativizing constructions} of key exchange from one-way functions, and oblivious transfer is stronger than key exchange. 
Since most cryptographic constructions are relativizing, this gives a strong negative result on constructing oblivious transfer from one-way functions in the classical setting.
}

As we have mentioned, it is classically impossible to realize commitments with
statistical hiding and statistical binding.
Does quantum physics overcome the barrier?
Unfortunately, it is already known that
both binding and hiding cannot be statistical at the same time even
in the quantum world~\cite{LC97,May97}. 
In fact, all known constructions of quantum commitments use at least
(quantum-secure) one-way functions~\cite{EC:DumMaySal00,EC:CreLegSal01,KO09,KO11,YWLQ15,Yan20,TCC:BitBra21}.

In this paper, we ask the following fundamental question:
\begin{center}
{\it 
Are one-way functions really necessary for commitments?}
\end{center}
It could be the case that in the quantum world commitments can be constructed 
from an assumption weaker than the existence of one-way functions.
This possibility is mentioned in previous works~\cite{C:BCKM21b,EC:GLSV21}, but no construction is provided.

Digital signatures~\cite{DifHel76} are other important primitives in cryptography. 
In a signature scheme, a secret key $sk$ and a public key $pk$ are generated. 
The secret key $sk$ is used to generate a signature $\sigma$ for a message $m$, and the public key $pk$ is used for the verification of the pair $(m,\sigma)$ of the message and the signature.
Any adversary who has $pk$ and can query the signing oracle many times cannot 
forge a signature $\sigma'$ for a message $m'$ which is not queried. In other words, $(m',\sigma')$ is not accepted by the verification algorithm except for an negligible probability.

Obviously, statistically-secure digital signatures are impossible, because an unbounded adversary who can access $pk$ and the verification algorithm can find a valid signature by a brute-force search.
In the classical world, it is known that the existence of digital signatures is equivalent to the existence of one-way functions.
In the quantum setting, on the other hand, digital signatures are not known to imply one-way functions.
Gottesman and Chuang introduced digital signatures with quantum public keys~\cite{Gottesman_Chuang},
but they considered information-theoretical security, and therefore the number of public keys should be bounded.
Our second fundamental question in this paper is the following:
\begin{center}
{\it Are digital signatures possible without one-way functions?}
\end{center}

\subsection{Our Results}
In this paper, we answer the above two fundamental questions affirmatively.
The first result of this paper is a construction of quantum commitments from pseudorandom quantum states generators (PRSGs)~\cite{C:JiLiuSon18,TCC:BraShm19,C:BraShm20}.
A PRSG is a quantum polynomial-time algorithm that, on input $k\in\bit^n$, outputs
an $m$-qubit state $|\phi_k\rangle$ such that
$|\phi_k\rangle^{\otimes t}$ over uniform random $k$ is computationally indistinguishable from
the same number of copies of Haar random states for any polynomial $t$.
(The formal definition of PRSGs is given in Definition~\ref{definition:PRS}.)

Our first result is stated as follows:\footnote{Our construction of commitments also satisfies perfect correctness, i.e.,
the probability that the honest receiver opens the correct bit committed by the honest sender is 1.}
\begin{theorem}
\label{mainresult1}
If a pseudorandom quantum states generator with $m\ge cn$ for a constant $c>1$ exists, then non-interactive quantum commitments (for classical messages) with computational hiding
and statistical binding exist.
\end{theorem}

In \cite{Kre21}, it is shown that PRSGs exist even if
$\BQP=\QMA$ relative to a quantum oracle.
If $\BQP=\QMA$, no quantum-secure classical cryptographic primitive exists,
because $\BQP=\QMA$ means $\NP\subseteq\BQP$. 
In particular, no quantum-secure one-way function exists.
Our Theorem~\ref{mainresult1} therefore shows that quantum
commitments can exist 
even if no quantum-secure classical cryptographic primitive exists.\footnote{It actually shows stronger things,
because $\BQP=\QMA$ also excludes the existence of some quantum-secure {\it quantum} cryptographic primitives
where honest algorithms are quantum.}
In particular, quantum commitments can exist even if no
quantum-secure one-way function exists.

As we will see later (Section~\ref{sec:commitments}), what we actually need is a weaker version of PRSGs
where only the computational indistinguishability of a single copy of $|\phi_k\rangle$ from the Haar random state is required.
We call such a weaker version of PRSGs {\it single-copy-secure PRSGs}.
(See Definition~\ref{definition:1PRS}. It is the $t=1$ version of Definition~\ref{definition:PRS}.)
Because a single copy of the Haar random state is equivalent to the maximally-mixed state,
the single-copy security means the computational indistinguishability from the
maximally-mixed state. 
It could be the case that the realization of single-copy-secure PRSGs is easier than
that of (multi-copy-secure) PRSGs.
(For more discussions, see Section~\ref{sec:PRSs}.)

Non-interactive commitments are a special type of commitments. 
(See Definition~\ref{definition:commitments}.)
In general, the sender and the receiver exchange many
rounds of messages during the commitment phase, but in non-interactive commitments, only a single message from
the sender to the receiver is enough for the commitment.
It is known that non-interactive quantum commitments (for classical messages)
are possible with (quantum-secure) one-way functions~\cite{YWLQ15},
while it is subject to a black-box barrier in the classical case~\cite{MP12}.

As the definition of binding, we choose a standard one, sum-binding~\cite{AC:Unruh16}, which roughly means that $p_0+p_1\le1+\negl(\lambda)$,
where $\negl$ is a negligible function, $\lambda$ is a security parameter,
and $p_0$ and $p_1$ are probabilities that the malicious sender makes the receiver open
0 and 1, respectively.
(The formal definition of statistical sum-binding is given in Definition~\ref{definition:sumbinding}.)

Our first result, Theorem~\ref{mainresult1}, that quantum commitments can be possible without one-way functions has  important
consequences in cryptography. 
It is known that quantum commitments imply the existence of quantum-secure zero-knowledge proofs (of knowledge) for all $\NP$ languages~\cite{FUYZ20} and 
quantum-secure oblivious transfer (and therefore multi-party computations (MPC))~\cite{C:BCKM21b,EC:GLSV21}.
Thus, those primitives can also exist even if $\BQP=\QMA$ (and in particular quantum-secure one-way functions do not exist)\footnote{
Indeed, \cite{C:BCKM21b} states as follows: 
``\emph{Moreover if in the future, new constructions of statistically binding, quantum computationally hiding commitments involving quantum communication are discovered based on assumptions weaker
than quantum-hard one-way functions, it would be possible to plug those into our protocol compilers to
obtain QOT.}"
} 
while classical constructions of them imply the existence of one-way functions.
For more details, see Appendix~\ref{sec:AQY_binding}.

We also remark that there is no known construction of PRSGs from weaker assumptions than the existence of one-way functions without oracles. 
Thus, our result should be understood as a theoretical evidence that quantum commitments can exist even if $\BQP=\QMA$ rather than a new concrete construction.
It is an interesting open problem to construct a PRSG from weaker assumptions than the existence of one-way functions without oracles. 
Such a construction immediately yields commitments (and more) by our result. 

One might ask the following question: can we remove (or improve) the condition of $m\ge cn$ with a constant $c>1$ in Theorem~\ref{mainresult1}?
The answer is no for single-copy-secure PRSGs, because if $m\le n$, there is a trivial construction of a single-copy-secure PRSG without any assumption:
$|\phi_k\rangle\coloneqq|k_1,...,k_m\rangle$ for any $k\in\bit^n$, where $k_j$ is the $j$th bit of $k$.
In fact, $\frac{1}{2^n}\sum_{k\in\bit^n}|\phi_k\rangle\langle\phi_k|=\frac{I^{\otimes m}}{2^m}$.
If quantum commitments were constructed from such a single-copy-secure PRSG, we could realize quantum commitments without any assumption, which is known to be impossible~\cite{LC97,May97}.
We note that 
\cite{Kre21} considers only the case when $m=n$, but it is clear that the result holds for $m\ge cn$ with constant $c>1$.

Finally, we do not know whether the opposite of Theorem~\ref{mainresult1} holds or not. Namely,
do quantum commitments imply PRSGs (or single-copy-secure PRSGs)? 
It is an interesting open problem.

Now let us move on to our second subject, namely, digital signatures.
Our second result in this paper is the following:
\begin{theorem}
\label{theorem:signature}
If a pseudorandom quantum states generator with $m\ge cn$ for a constant $c>1$ exists, then
one-time secure digital signatures with quantum public keys exist.
\end{theorem}

One-time security means that the adversary can query the signing oracle at most once. 
(See Definition~\ref{definition:signatures} and Definition~\ref{definition:one-time_security}.)
In the classical setting, it is known how to
construct many-time secure digital signatures from one-time secure digital signatures~\cite{C:Merkle89a}, 
but we do not know how to generalize our one-time secure quantum signature scheme
to a many-time secure one, because in our case public keys are quantum. It is an important open problem to construct many-time 
secure digital signatures from PRSGs.

Due to the oracle separation by \cite{Kre21}, Theorem~\ref{theorem:signature} means that (at least one-time secure) digital
signatures can exist even if no
quantum-secure classical cryptographic primitive exists.\footnote{Again, it also excludes some
quantum-secure {\it quantum} cryptographic primitives.}
In particular, (one-time secure) digital signatures can exist even if no
quantum-secure one-way function exists.

Our construction is similar to the ``quantum public key version" of the 
classical Lamport signature~\cite{DifHel76} by Gottesman and Chuang~\cite{Gottesman_Chuang}.
They consider information-theoretical security, and therefore the number of public keys should be bounded.
On the other hand, our construction from PRSGs allows unbounded polynomial number of public keys.
Quantum cryptography with quantum public keys are also studied in \cite{JC:KKNY12,Doliskani21}.

We do not know whether the condition, $m\ge cn$ with a constant $c>1$, can be improved or not in Theorem~\ref{theorem:signature}.
Although it is possible to construct PRSGs without this 
restriction~\cite{C:BraShm20}, this is satisfied in
the construction of \cite{Kre21}, and therefore enough for our purpose of
showing the existence of digital signatures without one-way functions.

As we will see later (Section~\ref{sec:signatures}), our construction of digital signatures is actually based on what we call {\it one-way quantum states generators (OWSGs)} (Definition~\ref{definition:onewayness}). 
Intuitively, we say that a quantum polynomial-time algorithm that outputs an $m$-qubit quantum state $|\phi_k\rangle$ on input $k\in\bit^n$
is a OWSG if it is hard to find, given polynomially many copies of $|\phi_k\rangle$ (with uniformly random $k$), an $n$-bit string $\sigma\in\bit^n$ such that $|\phi_\sigma\rangle$ is close to $|\phi_k\rangle$.
In other words, what we actually show is the following:
\begin{theorem}
\label{theorem:signature2}
If a one-way quantum states generator exists, then
one-time secure digital signatures with quantum public keys exist.
\end{theorem}
We show that a PRSG is a OWSG (Lemma~\ref{lemma:hiding2oneway}),
and therefore, Theorem~\ref{theorem:signature} is obtained as a corollary of Theorem~\ref{theorem:signature2}.
The concept of OWSGs itself seems to be of independent interest.
In particular, we do not know whether OWSGs imply PRSGs or not, which is an interesting open problem. 

Remember that for the construction of our commitment scheme
we use only single-copy-secure PRSGs.
Unlike our commitment scheme, on the other hand,
our signature scheme uses the security of PRSGs for an unbounded polynomial number of copies, 
because the number of copies decides the number of quantum public keys.
In other words, single-copy-secure PRSGs enable commitments but (multi-copy-secure) PRSGs enable signatures.
There could be therefore a kind of hierarchy in PRSGs for different numbers of copies, which seems to be an interesting future research subject.

\subsection{Technical Overviews}
Here we provide intuitive explanations of our constructions given in 
Section~\ref{sec:commitments} and Section~\ref{sec:signatures}.

\subsubsection{Commitments}
The basic idea of our construction of commitments
is, in some sense, a quantum
generalization of the classical Naor's commitment scheme~\cite{Nao91}. 

Let us recall Naor's construction. The receiver first samples uniformly random 
$\eta\leftarrow\{0,1\}^{3n}$, 
and sends it to the sender. The sender 
chooses a uniformly random seed $s\leftarrow\{0,1\}^n$, and
sends $G(s)\oplus \eta^b$ to the receiver,
where $G:\{0,1\}^n\to\{0,1\}^{3n}$ is a length-tripling pseudorandom generator,
and $b\in\{0,1\}$ is the bit to commit.
Hiding is clear: because the receiver does not know $s$, the receiver cannot distinguish
$G(s)$ and $G(s)\oplus \eta$.
The decommitment is $(b,s)$. The receiver can check whether the commitment is
$G(s)$ or $G(s)\oplus \eta$ from $s$.
Binding comes from the fact that if both 0 and 1 can be opened, 
there exist $s_0,s_1$ such that $G(s_0)=G(s_1)=\eta$. 
There are $2^{2n}$ such seeds, and therefore for a random $\eta$,
it is impossible except for $2^{-n}$ probability.

Our idea is to replace $G(s)$ with a pseudorandom state $|\phi_k\rangle$, and
to replace the addition of $\eta^b$ with the quantum one-time pad, which randomly applies Pauli $X$ and $Z$.
More precisely, the sender who wants to commit $b\in\bit$ generates the state
\begin{eqnarray*}
|\psi_b\rangle
\coloneqq\frac{1}{\sqrt{2^{2m+n}}}
\sum_{x,z\in\{0,1\}^m}
\sum_{k\in\{0,1\}^n}
|x,z,k\rangle_R\otimes P_{x,z}^b|\phi_k\rangle_C,
\end{eqnarray*}
and sends the register $C$ to the receiver,
where $P_{x,z}\coloneqq\bigotimes_{j=1}^mX_j^{x_j}Z_j^{z_j}$.
It is the commitment phase. At the end of the commitment phase, the receiver's state is 
$\rho_0\coloneqq\frac{1}{2^n}\sum_k|\phi_k\rangle\langle\phi_k|$ when $b=0$ and 
$\rho_1\coloneqq\frac{1}{2^n}\frac{1}{4^m}\sum_k\sum_{x,z}P_{x,z}^b|\phi_k\rangle\langle\phi_k|P_{x,z}^b$ when $b=1$.
By the security of single-copy-secure PRSGs, 
$\rho_0$
is computationally indistinguishable from the $m$-qubit maximally-mixed state $\frac{I^{\otimes m}}{2^m}$,
while $\rho_1=\frac{I^{\otimes m}}{2^m}$ due to the quantum one-time pad (Lemma~\ref{lemma:QOTP}).
The two states, $\rho_0$ and $\rho_1$, are therefore computationally indistinguishable,
which shows computational hiding.

For statistical sum-binding, we show that the fidelity between 
$\rho_0$ and $\rho_1$ is negligibly small.
It is intuitively understood as follows:
$\rho_0=\frac{1}{2^n}\sum_{k\in\bit^n}|\phi_k\rangle\langle\phi_k|$ has a support in at most $2^n$-dimensional space, 
while $\rho_1=\frac{I^{\otimes m}}{2^m}$ has a support in the entire $2^m$-dimensional space, where $m\ge cn$ with $c>1$,
and therefore the ``overlap" between $\rho_0$ and $\rho_1$ is small.

A detailed explanation of our construction of commitments and its security proof are given in Section~\ref{sec:commitments}.

\subsubsection{Digital Signatures}
Our construction of digital signatures is a quantum public key version of the classical Lamport signature.
The Lamport signature scheme is constructed from a one-way function.
For simplicity, let us explain the Lamport signature scheme for a single-bit message.
Let $f$ be a one-way function.
The secret key is $sk\coloneqq(sk_0,sk_1)$, where $sk_0,sk_1$ are uniform randomly chosen $n$-bit strings.
The public key is $pk\coloneqq(pk_0,pk_1)$, where $pk_0\coloneqq f(sk_0)$ and $pk_1\coloneqq f(sk_1)$.
The signature $\sigma$ for a message $m\in\bit$ is $sk_m$,
and the verification is to check whether $pk_m=f(\sigma)$.
Intuitively, the (one-time) security of this signature scheme comes from that of the one-way function $f$.

We consider the quantum public key version of it: $pk$ is a quantum state.
More precisely, our key generation algorithm chooses $k_0,k_1\leftarrow\bit^n$, and
runs $|\phi_{k_b}\rangle\leftarrow \StateGen(k_b)$ for $b\in\bit$,
where $\StateGen$ is a PRSG.\footnote{It is not necessarily a PRSG. Any OWSG
(Definition~\ref{definition:onewayness}) is enough. For details, see Section~\ref{sec:signatures}.}
It outputs $sk\coloneqq (sk_0,sk_1)$ and $pk\coloneqq (pk_0,pk_1)$, where
$sk_b\coloneqq k_b$ and $pk_b\coloneqq|\phi_{sk_b}\rangle$ for $b\in\bit$.
To sign a bit $m\in\bit$, the signing algorithm outputs the signature $\sigma\coloneqq sk_m$.
Given the message-signature pair $(m,\sigma)$,
the verification algorithm measures $pk_m$ with $\{|\phi_\sigma\rangle\langle\phi_\sigma|,I-|\phi_\sigma\rangle\langle\phi_\sigma|\}$
and accepts if and only if the result is $|\phi_\sigma\rangle\langle\phi_\sigma|$.

Intuitively, this signature scheme is one-time secure because
$sk_b$ cannot be obtained from $|\phi_{sk_b}\rangle^{\otimes t}$: If $sk_b$ is obtained,
$|\phi_{sk_b}\rangle^{\otimes t}$ can be distinguished from Haar random states, which contradicts the security of PRSGs.
In order to formalize this intuition, we introduce what we call OWSGs (Definition~\ref{definition:onewayness}), and show that
PRSGs imply OWSGs (Lemma~\ref{lemma:hiding2oneway}).
For details of our construction of digital signatures and its security proof, see Section~\ref{sec:signatures}.

\subsection{Concurrent Work}
Few days after the first version of this paper was made online,
a concurrent work~\cite{AQY21} appeared.
The concurrent work also constructs commitments from PRSGs. We give comparisons between our and their results.
\begin{enumerate}
    \item For achieving the security level of $O(2^{-n})$ for binding, they rely on PRSGs with $m=2\log n+\omega(\log\log n)$ and any $t$,
    or $m=7n$ and $t=1$.
    On the other hand, we rely on PRSGs with $m=3n$ and $t=1$.
    Thus, the required parameters seem incomparable though we cannot simply compare them due to the difference of definitions of binding. (See also Appendix~\ref{sec:AQY_binding}.)
    
\item Our scheme is non-interactive whereas theirs is interactive though we believe that their scheme can also be made non-interactive by a similar technique to ours. 
\item They consider a more general definition of PRSGs than us that allows the state generation algorithm to sometimes fail. We do not take this into account since we can rely on PRSGs of \cite{Kre21} whose state generation never fails for our primary goal to show that commitments and digital signatures can exist even if one-way functions do not exist. 
Moreover, the failure probability has to be anyway negligibly small due to the security of PRSGs, and
therefore it would be simpler to ignore the failure.
\end{enumerate}
Besides commitments,  
the result on digital signatures is unique to this paper.
On the other hand, \cite{AQY21} contains results that are not covered in this paper such as pseudorandom function-like states and symmetric key encryption.
We remark that our result on digital signatures was added a few days after the initial version of~\cite{AQY21} was made online, but the result was obtained independently, and there is no overlap with~\cite{AQY21} in this part.  

Though most part of this work was done
independently of~\cite{AQY21}, there are two part where we revised the paper based on~\cite{AQY21}. 
The first is the definition of PRSGs.  As pointed out in the initial version of~\cite{AQY21}, the initial version of this work implicitly assumed that PRSGs do not use any ancillary qubits, which is a very strong restriction. However, we found that all of our results can be based on PRSGs that use ancillary qubits with just notational adaptations.
Thus, we regard this as a notational level issue and fixed it. 

The second is the connection to oblivious transfer and MPC explained in Appendix~\ref{sec:AQY_binding}. In the initial version of this work, we only mentioned the idea of using the techniques of \cite{FUYZ20} to instantiate oblivious transfer and MPC of~\cite{C:BCKM21b} based on quantum commitments. On the other hand, \cite{AQY21} shows it assuming that the base quantum commitment satisfies a newly introduced definition of statistical binding property, which we call AQY-binding. 
Interestingly, we found that it is already implicitly shown in~\cite{FUYZ20} that the sum-binding implies AQY-binding. As a result, our commitment scheme can also be used to instantiate oblivious transfer and MPC of~\cite{C:BCKM21b}. See Appendix~\ref{sec:AQY_binding} for more detail.

\section{Preliminaries}
In this section, we provide preliminaries.

\subsection{Basic Notations}
We use standard notations in quantum information. For example, $I\coloneqq|0\rangle\langle0|+|1\rangle\langle1|$ is the two-dimensional identity operator.
For notational simplicity, we sometimes write the $n$-qubit identity operator $I^{\otimes n}$ just $I$ when it is clear from the context.
$X,Y,Z$ are Pauli operators.
$X_j$ means the Pauli $X$ operator that acts on the $j$th qubit.
Let $\rho_{AB}$ be a quantum state over the subsystems $A$ and $B$.
Then $\mbox{Tr}_A(\rho_{AB})$ is the partial trace of $\rho_{AB}$ over subsystem $A$.
For $n$-bit strings $x,z\in\bit^n$, $X^x\coloneqq \bigotimes_{j=1}^nX_j^{x_j}$
and $Z^z\coloneqq \bigotimes_{j=1}^nZ_j^{z_j}$,
where $x_j$ and $z_j$ are the $j$th bit of $x$ and $z$, respectively.

We also use standard notations in cryptography.
A function $f$ is negligible if for all constant $c>0$, $f(\lambda)<\lambda^{-c}$ for large enough $\lambda$.
QPT and PPT stand for quantum polynomial time and (classical) probabilistic polynomial time, respectively.
$k\leftarrow\bit^n$ means that $k$ is sampled from $\bit^n$ uniformly at random.
For an algorithm $\cA$, $\cA(\xi)\to \eta$ means that the algorithm outputs $\eta$ on input $\xi$.

In this paper, we use the following lemma. It can be shown by a straightforward calculation.
\begin{lemma}[Quantum one-time pad]
\label{lemma:QOTP}
For any $m$-qubit state $\rho$,
\begin{eqnarray*}
\frac{1}{4^m}\sum_{x\in\bit^m}\sum_{z\in\bit^m}X^xZ^z\rho Z^z X^x=\frac{I^{\otimes m}}{2^m}.
\end{eqnarray*}
\end{lemma}

\subsection{Pseudorandom Quantum States Generators}
\label{sec:PRSs}
Let us review pseudorandom quantum states generators (PRSGs)~\cite{C:JiLiuSon18,TCC:BraShm19,C:BraShm20}.
The definition of PRSGs is given as follows.

\begin{definition}[Pseudorandom quantum states generators (PRSGs)~\cite{C:JiLiuSon18,TCC:BraShm19,C:BraShm20}]
\label{definition:PRS}
A pseudorandom quantum states generator (PRSG) is a QPT
algorithm $\StateGen$ that, on input $k\in\{0,1\}^n$,
outputs an $m$-qubit quantum state $|\phi_k\rangle$.
As the security, we require the following: 
for any polynomial $t$ and any non-uniform QPT adversary $\cA$,  
there exists a negligible function $\negl$
such that for all $n$,
\begin{eqnarray*}
\Big|\Pr_{k\leftarrow \{0,1\}^n}\Big[\cA(|\phi_k\rangle^{\otimes t(n)})\to1\Big]
-\Pr_{|\psi\rangle\leftarrow \mu_m}\Big[\cA(|\psi\rangle^{\otimes t(n)})\to1\Big]\Big|
\le \negl(n),
\end{eqnarray*}
where $\mu_m$ is the Haar measure on $m$-qubit states.
\end{definition}

\begin{remark}
\label{remark:pure}
In the most general case, $\StateGen$ is the following QPT algorithm: 
given an input $k\in\bit^n$, it first computes a classical description of
a unitary quantum circuit $U_k$, and next applies $U_k$ on the all zero state $|0...0\rangle$
to generate $|\Phi_k\rangle_{AB}\coloneqq U_k|0...0\rangle$.
It finally outputs the $m$-qubit state $\rho_k\coloneqq \mbox{Tr}_B(|\Phi_k\rangle\langle\Phi_k|_{AB})$.
However, $\rho_k$ is, on average, almost pure,
because otherwise the security is broken by a QPT adversary who runs the SWAP test on two copies.\footnote{
Let us consider an adversary $\cA$ that runs the SWAP test on two copies of the received state and outputs the result of the SWAP test.
When $\rho_k^{\otimes t}$ is sent with uniformly random $k$, the probability that $\cA$ outputs 1 is $(1+\frac{1}{2^n}\sum_k\mbox{Tr}(\rho_k^2))/2$.
When the $t$ copies of Haar random states $|\psi\rangle^{\otimes t}$ is sent, the probability that $\cA$ outputs 1 is 1.
For the security, $|(1+\frac{1}{2^n}\sum_k\mbox{Tr}(\rho_k^2))/2-1|\le \negl(\lambda)$ has to be satisfied, which means 
the expected purity of $\rho_k$, $\frac{1}{2^n}\sum_k\mbox{Tr}(\rho_k^2)$, has to be negligibly close to 1.}
In this paper, for simplicity, we assume that $\rho_k$ is pure, and denote it by $|\phi_k\rangle$.
The same results hold even if it is negligibly close to pure.
What $\StateGen$ generates is therefore
$|\Phi_k\rangle_{AB}=|\phi_k\rangle_A\otimes |\eta_k\rangle_B$
with an ancilla state $|\eta_k\rangle$.
In this paper, for simplicity, we assume that there is no ancilla state in the final state generated by $\StateGen$, but
actually the same results hold even if ancilla states exist. (See Section~\ref{sec:commitments} and Section~\ref{sec:signatures}.)
Moreover, \cite{Kre21} considers the case with pure outputs and no ancilla state, and therefore 
restricting to the pure and no ancilla case is enough for
our purpose of showing the existence of commitments and digital signatures without one-way functions.
\end{remark}

Interestingly, what we actually need for our construction of commitments (Section~\ref{sec:commitments}) is a weaker version of PRSGs
where the security is satisfied only for $t=1$. 
We call them {\it single-copy-secure PRSGs}:
\begin{definition}[Single-copy-secure PRSGs]
\label{definition:1PRS}
A single-copy-secure pseudorandom quantum states generator (PRSG) is a QPT
algorithm $\StateGen$ that, on input $k\in\{0,1\}^n$,
outputs an $m$-qubit quantum state $|\phi_k\rangle$.
As the security, we require the following: 
for any non-uniform QPT adversary $\cA$,  
there exists a negligible function $\negl$
such that for all $n$,
\begin{eqnarray*}
\Big|\Pr_{k\leftarrow \{0,1\}^n}\Big[\cA(|\phi_k\rangle)\to1\Big]
-\Pr_{|\psi\rangle\leftarrow \mu_m}\Big[\cA(|\psi\rangle)\to1\Big]\Big|
\le \negl(n),
\end{eqnarray*}
where $\mu_m$ is the Haar measure on $m$-qubit states.
\end{definition}

\begin{remark}
Because a single copy of an $m$-qubit state sampled Haar randomly is equivalent to the $m$-qubit maximally-mixed state, $\frac{I^{\otimes m}}{2^m}$,
the security of single-copy-secure PRSGs is in fact
the computational indistinguishability of
a single copy of $|\phi_k\rangle$ from $\frac{I^{\otimes m}}{2^m}$.
\end{remark}

\begin{remark}
As we have explained in Remark~\ref{remark:pure}, in the definition of PRSGs (Definition~\ref{definition:PRS}), the output state $\rho_k$ of $\StateGen$ has to be negligibly close to pure (on average).
When we consider single-copy-secure PRSGs (Definition~\ref{definition:1PRS}), on the other hand,
the SWAP-test attack does not work because only a single copy is available to
adversaries. In fact, there is a trivial construction of a single-copy-secure PRSG whose output is not
pure:
$\StateGen(k)\to\frac{I^{\otimes m}}{2^m}$ for all $k\in\bit^n$.
We therefore assume that the output of $\StateGen$ is pure, i.e., $\rho_k=|\phi_k\rangle\langle\phi_k|$, when we consider single-copy-secure PRSGs.
\end{remark}

\begin{remark}
It could be the case that single-copy-secure PRSGs are easier to realize than  
(multi-copy-secure) PRSGs.
In fact, the security proofs of the constructions of \cite{C:JiLiuSon18,TCC:BraShm19}
are simpler for $t=1$. Furthermore, there is a simple construction of a single-copy-secure PRSG by
using a pseudorandom generator $G:\{0,1\}^n\to\{0,1\}^m$. 
In fact, we have only to take $|\phi_k\rangle=|G(k)\rangle$ for all $k\in\bit^n$.
\end{remark}

\begin{remark}
One might think that a single-copy-secure PRSG with $m\ge n+1$ is a pseudorandom generator (PRG), because if the $m$-qubit state $|\phi_k\rangle$
is measured in the computational basis, the probability distribution of the measurement results
is computationally indistinguishable from that (i.e., the $m$-bit uniform distribution)
obtained when the $m$-qubit maximally mixed state $\frac{I^{\otimes m}}{2^m}$ is measured in the computational basis.
It is, however, strange because if it was true then the existence of single-copy-secure PRSGs implies the existence of PRGs, 
which contradicts \cite{Kre21}.
The point is that measuring $|\phi_k\rangle$ in the computational basis does not work as PRGs because the output is not deterministically obtained. 
(Remember that PRGs are deterministic algorithms.)
\end{remark}

\section{Commitments}
\label{sec:commitments}
In this section, we provide our construction of commitments, and show its security.

\subsection{Definition}
Let us first give a formal definition of non-interactive quantum commitments.

\begin{definition}[Non-interactive quantum commitments (Syntax)]
\label{definition:commitments}
A non-interactive quantum commitment scheme is the following protocol.
\begin{itemize}
\item
{\bf Commit phase:}
Let $b\in\{0,1\}$ be the bit to commit. 
The sender generates a quantum state $|\psi_b\rangle_{RC}$ on registers $R$ and $C$, and 
sends the register $C$ to the receiver.
The states $\{|\psi_b\rangle\}_{b\in\{0,1\}}$ can be generated in quantum polynomial-time from the all zero state.
\item
{\bf Reveal phase:}
The sender sends $b$ and the register $R$ to the receiver.
The receiver does the measurement $\{|\psi_b\rangle\langle\psi_b|,I-|\psi_b\rangle\langle\psi_b|\}$
on the registers $R$ and $C$. If the result is $|\psi_b\rangle\langle\psi_b|$,
the receiver outputs $b$. Otherwise, the receiver outputs $\bot$.
Because 
$\{|\psi_b\rangle\}_{b\in\{0,1\}}$ can be generated in quantum polynomial-time from the all zero state,
the measurement $\{|\psi_b\rangle\langle\psi_b|,I-|\psi_b\rangle\langle\psi_b|\}$
can be implemented efficiently.
\end{itemize}
\end{definition}

The perfect correctness is defined as follows:
\begin{definition}[Perfect correctness]
\label{definition:commitments_correctness}
A commitment scheme satisfies perfect correctness if the following is satisfied:
when the honest sender commits $b\in\bit$, the probability that the honest receiver opens $b$
is 1.
\end{definition}

The computational hiding is defined as follows:
\begin{definition}[Computational hiding]
Let us consider the following security game, $\mathsf{Exp}(b)$, with the parameter $b\in\{0,1\}$
between a challenger $\cC$ and a QPT adversary $\cA$.
\begin{enumerate}
\item
$\cC$ generates 
$|\psi_b\rangle_{RC}$ 
and sends the register $C$ to $\cA$.
\item
$\cA$ outputs $b'\in\{0,1\}$,
which is the output of the experiment.
\end{enumerate}
We say that a non-interactive quantum commitment scheme is computationally hiding if
for any QPT adversary $\cA$ there exists a negligible function $\negl$ such that,
\begin{eqnarray*}
\left|\Pr[\mathsf{Exp}(0)=1]-\Pr[\mathsf{Exp}(1)=1]\right|\le \negl(\lambda).
\end{eqnarray*}
\end{definition}

As the definition of binding, we consider sum-binding~\cite{AC:Unruh16} that is defined as follows:

\begin{definition}[Statistical sum-binding]
\label{definition:sumbinding}
Let us consider the following security game between a challenger $\cC$ and an unbounded adversary $\cA$:
\begin{enumerate}
\item
$\cA$ generates a quantum state $|\Psi\rangle_{ERC}$ on the three registers $E$, $R$, and $C$.
\item
$\cA$ sends the register $C$ to $\cC$, which is the commitment.
\item
If $\cA$ wants to make $\cC$ open $b\in\bit$, $\cA$ applies a unitary
$U_{ER}^{(b)}$ on the registers $E$ and $R$. $\cA$ sends $b$ and the register $R$ to $\cC$.
\item
$\cC$ does the measurement $\{|\psi_b\rangle\langle\psi_b|,I-|\psi_b\rangle\langle\psi_b|\}$ on
the registers $R$ and $C$. If the result $|\psi_b\rangle\langle\psi_b|$ is obtained, 
$\cC$ accepts $b$. Otherwise, $\cC$ outputs $\bot$.
\end{enumerate}
Let $p_b$ be the probability that $\cA$ makes $\cC$ open $b\in\{0,1\}$:
\begin{eqnarray*}
p_b
\coloneqq
\langle \psi_b|_{RC}{\rm Tr}_E(U_{ER}^{(b)}|\Psi\rangle\langle\Psi|_{ERC}U_{ER}^{(b)\dagger})|\psi_b\rangle_{RC}.
\end{eqnarray*}
We say that the commitment scheme is statistical sum-binding if
for any unbounded $\cA$ there exists a negligible function $\negl$
such that
\begin{eqnarray*}
p_0+p_1\le1+\negl(\lambda). 
\end{eqnarray*}
\end{definition}

\subsection{Construction}
Let us explain our construction of commitments.\footnote{
Another example of constructions is 
$|\psi_0\rangle=\sum_{k\in\{0,1\}^n}|k\rangle|\phi_k\rangle$
and
$|\psi_1\rangle=\sum_{r\in\{0,1\}^m}|r\rangle|r\rangle$.
We have chosen the one we have explained, because
the analogy to Naor's commitment scheme is clearer.} 
Let $\StateGen$ be a single-copy-secure PRSG that, on input $k\in\bit^n$,
outputs an $m$-qubit state $|\phi_k\rangle$.
The commit phase is the following.
\begin{enumerate}
\item
Let $b\in\{0,1\}$ be the bit to commit.
The sender generates 
\begin{eqnarray*}
|\psi_b\rangle
\coloneqq\frac{1}{\sqrt{2^{2m+n}}}
\sum_{x,z\in\{0,1\}^m}
\sum_{k\in\{0,1\}^n}
|x,z,k\rangle_R\otimes P_{x,z}^b|\phi_k\rangle_C,
\end{eqnarray*}
and sends the register $C$ to the receiver,
where $P_{x,z}\coloneqq\bigotimes_{j=1}^mX_j^{x_j}Z_j^{z_j}$.
\end{enumerate}
The reveal phase is the following.
\begin{enumerate}
\item
The sender sends the register $R$ and the bit $b$ to the receiver.
\item
The receiver measures the state with 
$\{|\psi_b\rangle\langle\psi_b|,I-|\psi_b\rangle\langle\psi_b|\}$. 
If the result is $|\psi_b\rangle\langle\psi_b|$, 
the receiver outputs $b$. Otherwise, the receiver outputs $\bot$.
(Note that such a measurement can be done efficiently: first apply $V_b^\dagger$ such that
$|\psi_b\rangle=V_b|0...0\rangle$, and then measure all qubits in the computational basis to see
whether all results are zero or not.)
\end{enumerate}
It is obvious that this construction satisfies perfect correctness (Definition~\ref{definition:commitments_correctness}).

\begin{remark}
Note that if we slightly modify the above construction, the communication in the reveal phase can be classical.
In fact, we can show it for general settings.
We will provide a detailed explanation of it in Appendix~\ref{sec:classical_opening}. Here, we give an intuitive argument.
In general non-interactive quantum commitments (Definition~\ref{definition:commitments}), the sender who wants to commit $b\in\bit$ generates a certain state $|\psi_b\rangle_{RC}$ on the registers $R$ and $C$,
and sends the register $C$ to the receiver, which is the commit phase.
In the reveal phase, $b$ and the register $R$ are sent to the receiver. The receiver runs the verification algorithm on the registers $R$ and $C$.
Let us modify it as follows. 
In the commit phase, the sender chooses uniform random $x,z\leftarrow\bit^{|R|}$ and applies 
$\bigotimes_{j=1}^{|R|}X_j^{x_j}Z_j^{z_j}$ 
on the register $R$ of $|\psi_b\rangle_{RC}$, where $|R|$ is the number of qubits in the register $R$.
The sender then sends both the registers $R$ and $C$ to the receiver. It ends the commit phase.
In the reveal phase, the sender sends the bit $b$ to open and $(x,z)$ to the receiver. The receiver applies 
$\bigotimes_{j=1}^{|R|}X_j^{x_j}Z_j^{z_j}$ on the register $R$ and runs the original verification algorithm.
Hiding is clear because the register $R$ is traced out to the receiver before the reveal phase due to the quantum one-time pad.
Binding is also easy to understand: Assume a malicious sender of the modified scheme can break binding. Then, we can construct a malicious sender that breaks binding of the original scheme,
because the malicious sender of the original scheme can simulate the malicious sender of the modified scheme. 
\end{remark}

\begin{remark}
We also note that our construction of commitments can be extended to more general cases
where ancilla qubits are used in PRSGs. Let us consider a more general PRSG
that generates $|\phi_k\rangle\otimes|\eta_k\rangle$ and outputs $|\phi_k\rangle$, where
$|\eta_k\rangle$ is an ancilla state.
In that case, hiding and binding hold if we replace $|\psi_b\rangle$ with
\begin{eqnarray*}
\frac{1}{\sqrt{2^{2m+n}}}
\sum_{x,z\in\{0,1\}^m}
\sum_{k\in\{0,1\}^n}
(|x,z,k\rangle\otimes|\eta_k\rangle)_R
\otimes P_{x,z}^b|\phi_k\rangle_C.
\end{eqnarray*}
\end{remark}

\subsection{Computational Hiding}
We show computational hiding of our construction.

\begin{theorem}[Computational hiding]
\label{theorem:hiding}
Our construction satisfies computational hiding.
\end{theorem}

\begin{proof}[Proof of Theorem~\ref{theorem:hiding}]
Let us consider the following security game, 
$\hyb_0(b)$, which is the same as the original
experiment.
\begin{enumerate}
\item
The challenger $\cC$ generates
\begin{eqnarray*}
|\psi_b\rangle=\frac{1}{\sqrt{2^{2m+n}}}
\sum_{x,z\in\{0,1\}^m}
\sum_{k\in\{0,1\}^n}
|x,z,k\rangle_R\otimes P_{x,z}^b|\phi_k\rangle_C,
\end{eqnarray*}
and sends
the register $C$ to the adversary $\cA$,
where $P_{x,z}\coloneqq\bigotimes_{j=1}^mX_j^{x_j}Z_j^{z_j}$.
\item
$\cA$ outputs $b'\in\{0,1\}$, which is the output of this hybrid.
\end{enumerate}

Let us define $\hyb_1(b)$ as follows:
\begin{enumerate}
\item
If $b=0$, $\cC$ 
chooses a Haar random $m$-qubit state $|\psi\rangle\leftarrow\mu_m$, and sends it to $\cA$.
If $b=1$, $\cC$ generates
$|\psi_1\rangle_{RC}$
and sends
the register $C$ to $\cA$.
\item
$\cA$ outputs $b'\in\{0,1\}$, which is the output of this hybrid.
\end{enumerate}

\begin{lemma}
\label{lemma:01}
\begin{eqnarray*}
\left|\Pr[\hyb_0(b)=1]-\Pr[\hyb_1(b)=1]\right|\le \negl(\lambda)
\end{eqnarray*}
for each $b\in\{0,1\}$.
\end{lemma}

\begin{proof}[Proof of Lemma~\ref{lemma:01}]
It is clear that 
\begin{eqnarray*}
\Prob[\hyb_0(1)=1]=\Prob[\hyb_1(1)=1].
\end{eqnarray*}
Let us show 
\begin{eqnarray*}
|\Prob[\hyb_0(0)=1]-\Prob[\hyb_1(0)=1]|\le \negl(\lambda).
\end{eqnarray*}
To show it, assume that 
\begin{eqnarray*}
|\Prob[\hyb_0(0)=1]-\Prob[\hyb_1(0)=1]|
\end{eqnarray*}
is non-negligible.
Then, we can construct an adversary $\cA'$ that breaks the security of PRSGs as follows.
Let $b''\in\{0,1\}$ be the parameter of the security game of PRSGs.
\begin{enumerate}
\item
The challenger $\cC'$ of the security game of PRSGs sends $\cA'$
the state $|\phi_k\rangle$ with uniform random $k$ if $b''=0$ 
and a Haar random state $|\psi\rangle\leftarrow \mu_m$ if $b''=1$.
\item
$\cA'$ sends the received state to $\cA$.
\item
$\cA'$ outputs the output of $\cA$. 
\end{enumerate}
If $b''=0$, it simulates $\hyb_0(0)$.
If $b''=1$, it simulates $\hyb_1(0)$.
Therefore, $\cA'$ breaks the security of PRSGs. 
\end{proof}

Let us define $\hyb_2(b)$ as follows:
\begin{enumerate}
\item
The challenger $\cC$
chooses a Haar random $m$-qubit state $|\psi\rangle\leftarrow\mu_m$, and sends it to the adversary.
\item
The adversary outputs $b'\in\{0,1\}$, which is the output of this hybrid.
\end{enumerate}

\begin{lemma}
\label{lemma:12}
\begin{eqnarray*}
|\Pr[\hyb_1(b)=1]-\Pr[\hyb_2(b)=1]|\le \negl(\lambda)
\end{eqnarray*}
for each $b\in\{0,1\}$.
\end{lemma}

\begin{proof}[Proof of Lemma~\ref{lemma:12}]
\begin{eqnarray*}
\Prob[\hyb_1(0)=1]=\Pr[\hyb_2(0)=1]
\end{eqnarray*}
is clear.
Let us show 
\begin{eqnarray*}
|\Pr[\hyb_1(1)=1]-\Pr[\hyb_2(1)=1]|\le \negl(\lambda).
\end{eqnarray*}
To show it, assume that 
\begin{eqnarray*}
|\Pr[\hyb_1(1)=1]-\Pr[\hyb_2(1)=1]|
\end{eqnarray*}
is non-negligible.
Then, we can construct an adversary $\cA'$ that breaks the security of PRSGs as follows.
Let $b''\in\{0,1\}$ be the parameter of the security game of PRSGs.
\begin{enumerate}
\item
The challenger $\cC'$ of the security game of PRSGs sends $\cA'$
the state $|\phi_k\rangle$ with uniform random $k$ if $b''=0$ and 
a Haar random state $|\psi\rangle\leftarrow \mu_m$ if $b''=1$.
\item
$\cA'$ applies $X^xZ^z$ with uniform random $x,z\leftarrow\{0,1\}^m$, 
and sends the state to $\cA$.
\item
$\cA'$ outputs the output of $\cA$. 
\end{enumerate}
If $b''=0$, it simulates $\hyb_1(1)$.
If $b''=1$, it simulates $\hyb_2(1)$.
Therefore, $\cA'$ breaks the security of PRSGs. 
\end{proof}

It is obvious that 
\begin{eqnarray*}
\Pr[\hyb_2(0)=1]=\Pr[\hyb_2(1)=1].
\end{eqnarray*}
Therefore, from Lemma~\ref{lemma:01} and Lemma~\ref{lemma:12}, we conclude
\begin{eqnarray*}
|\Pr[\hyb_0(0)=1]-\Pr[\hyb_0(1)=1]|\le\negl(\lambda),
\end{eqnarray*}
which shows Theorem~\ref{theorem:hiding}.
\end{proof}

\subsection{Statistical Binding}

Let us show that our construction satisfies statistical sum-binding.

\begin{theorem}[Statistical sum-binding]
\label{theorem:binding}
Our construction satisfies statistical sum-binding.
\end{theorem}

\begin{proof}[Proof of Theorem~\ref{theorem:binding}]
Let
\begin{eqnarray*}
F(\rho,\sigma):=\Big(\mbox{Tr}\sqrt{\sqrt{\sigma}\rho\sqrt{\sigma}}\Big)^2
\end{eqnarray*}
be the fidelity between two states $\rho$ and $\sigma$.
Then we have
\begin{eqnarray*}
p_b
&=&
\langle \psi_b|_{RC}\mbox{Tr}_E(U_{ER}^{(b)}|\Psi\rangle\langle\Psi|_{ERC}U_{ER}^{(b)\dagger})|\psi_b\rangle_{RC}\\
&=&
F\Big(|\psi_b\rangle_{RC},\mbox{Tr}_E(U_{ER}^{(b)}|\Psi\rangle\langle\Psi|_{ERC}U_{ER}^{(b)\dagger})\Big)\\
&\le&
F\Big(\mbox{Tr}_R(|\psi_b\rangle\langle\psi_b|_{RC}),\mbox{Tr}_{RE}(U_{ER}^{(b)}|\Psi\rangle\langle\Psi|_{ERC}
U_{ER}^{(b)\dagger})\Big)\\
&=&
F\Big(\mbox{Tr}_R(|\psi_b\rangle\langle\psi_b|_{RC}),\mbox{Tr}_{RE}(|\Psi\rangle\langle\Psi|_{ERC})\Big).
\end{eqnarray*}
Here, we have used the facts that
if $\sigma=|\sigma\rangle\langle\sigma|$,
$F(\rho,\sigma)=\langle\sigma|\rho|\sigma\rangle$,
and that for any bipartite states $\rho_{AB},\sigma_{AB}$,
$F(\rho_{AB},\sigma_{AB})\le F(\rho_A,\sigma_A)$,
where $\rho_A=\mbox{Tr}_B(\rho_{AB})$
and $\sigma_A=\mbox{Tr}_B(\sigma_{AB})$.

Therefore,
\begin{eqnarray*}
p_0+p_1
&\le&
1+\sqrt{F\Big(\mbox{Tr}_R(|\psi_0\rangle\langle\psi_0|_{RC}),\mbox{Tr}_R(|\psi_1\rangle\langle\psi_1|_{RC})\Big)}\\
&=&
1+\sqrt{F\Big(
\frac{1}{2^n}\sum_k|\phi_k\rangle\langle\phi_k|, 
\frac{1}{2^{2m}}\frac{1}{2^n}\sum_{x,z}\sum_kX^xZ^z|\phi_k\rangle\langle\phi_k|X^xZ^z
\Big)}\\
&=&
1+\sqrt{F\Big(
\frac{1}{2^n}\sum_k|\phi_k\rangle\langle\phi_k|, 
\frac{I^{\otimes m}}{2^m}
\Big)}\\
&=&1+\Big\|\sum_{i=1}^\xi\sqrt{\lambda_i}\frac{1}{\sqrt{2^m}}|\lambda_i\rangle\langle\lambda_i|\Big\|_1\\
&=&1+\sum_{i=1}^\xi\sqrt{\lambda_i}\frac{1}{\sqrt{2^m}}\\
&\le&1+\sqrt{\sum_{i=1}^\xi\lambda_i}\sqrt{\sum_{i=1}^\xi\frac{1}{2^m}}\\
&\le&1+\sqrt{\frac{2^n}{2^m}}\\ 
&\le&1+\frac{1}{\sqrt{2^{(c-1)n}}}.
\end{eqnarray*}
In the first inequality, we
have used the fact that for any states $\rho,\sigma,\xi$,
\begin{eqnarray*}
F(\rho,\xi)+F(\sigma,\xi)\le 1+\sqrt{F(\rho,\sigma)}
\end{eqnarray*}
is satisfied~\cite{NayakShor}.
In the fourth equality,
$\sum_{i=1}^\xi\lambda_i|\lambda_i\rangle\langle\lambda_i|$
is the diagonalization of $\frac{1}{2^n}\sum_k|\phi_k\rangle\langle\phi_k|$. 
In the sixth inequality, we have used 
Cauchy–Schwarz inequality. 
In the seventh inequality, we have used $\xi\le2^n$.
In the last inequality, we have used 
$m\ge cn$ for a constant $c>1$. 
\end{proof}

\section{Digital Signatures}
\label{sec:signatures}
In this section, we provide our construction of digital signatures and show its security.
For that goal, we first define OWSGs (Definition~\ref{definition:onewayness}), 
and show that PRSGs imply OWSGs (Lemma~\ref{lemma:hiding2oneway}).

\subsection{One-way Quantum States Generators}
For the construction of our signature scheme, we introduce OWSGs, which are defined as follows:

\begin{definition}[One-way quantum states generators (OWSGs)]
\label{definition:onewayness}
Let $G$ be a QPT algorithm that, on input $k\in\bit^n$, outputs a quantum state $|\phi_k\rangle$.
Let us consider the following security game, $\mathsf{Exp}$, between a challenger $\cC$
and a QPT adversary $\cA$:
\begin{enumerate}
\item
$\cC$ chooses $k\leftarrow\bit^n$.
\item
$\cC$ runs $|\phi_k\rangle\leftarrow G(k)$ $t+1$ times.
\item
$\cC$ sends $|\phi_k\rangle^{\otimes t}$ to $\cA$.
\item
$\cA$ sends $\sigma\in\bit^n$ to $\cC$. 
\item
$\cC$ measures $|\phi_k\rangle$ with $\{|\phi_\sigma\rangle\langle\phi_\sigma|,I-|\phi_\sigma\rangle\langle\phi_\sigma|\}$.
If the result is $|\phi_\sigma\rangle\langle\phi_\sigma|$, the output of the experiment is 1.
Otherwise, the output of the experiment is 0.
\end{enumerate}
We say that $G$ is a one-way quantum states generator (OWSG) if
for any $t=poly(n)$ and for any QPT adversary $\cA$ there exists a negligible function $\negl$ such that
\begin{eqnarray*}
\Pr[\mathsf{Exp}=1]\le\negl(n).
\end{eqnarray*}
\end{definition}

\begin{remark}
Note that another natural definition of one-wayness is that
given $|\phi_k\rangle^{\otimes t}$ it is hard to find $k$.
However, as we will see later, it is not useful for our construction of digital signatures.
\end{remark}

\begin{remark}
The most general form of $G$ is as follows: on input $k\in\bit^n$, it computes a classical description of a unitary
quantum circuit $U_k$, and applies $U_k$
on $|0...0\rangle$ to generate $|\Phi_k\rangle_{AB}\coloneqq U_k|0...0\rangle$, and outputs
$\rho_k\coloneqq\mbox{Tr}_B(|\Phi_k\rangle\langle\Phi_k|_{AB})$.
However, because $\rho_k$ plays the role of a public key in our construction of digital signatures, we assume that $\rho_k$ is pure.
(It is not natural if public keys and secret keys are entangled.)
In that case, $U_k|0...0\rangle=|\phi_k\rangle_A\otimes|\eta_k\rangle_B$, where $|\eta_k\rangle$
is an ancilla state. For simplicity, we assume that there is no ancilla state:
$U_k|0...0\rangle=|\phi_k\rangle$.
In that case, the measurement $\{|\phi_\sigma\rangle\langle\phi_\sigma|,I-|\phi_\sigma\rangle\langle\phi_\sigma|\}$ 
by the challenger in Definition~\ref{definition:onewayness} can be done as follows: the challenger first applies $U_\sigma^\dagger$ on the state and then measures
all qubits in the comptuational basis. The all zero measurement result corresponds
to $|\phi_\sigma\rangle\langle\phi_\sigma|$ and other results correspond to $I-|\phi_\sigma\rangle\langle\phi_\sigma|$.
Even if ancilla states exist, however, the same result holds.
In that case, the verification of the challenger in Definition~\ref{definition:onewayness} is modified as follows: given $\sigma$, it generates $U_\sigma|0...0\rangle=|\phi_\sigma\rangle_A\otimes|\eta_\sigma\rangle_B$ to obtain $|\eta_\sigma\rangle$, 
applies $U_\sigma^\dagger$ on $|\phi_k\rangle\otimes|\eta_\sigma\rangle$,
and measures all qubits in the computational basis. 
If the result is all zero, it accepts, i.e., the output of the experiment is 1.
Otherwise, it rejects.
\end{remark}

We can show the following:

\begin{lemma}[PRSGs imply OWSGs]
\label{lemma:hiding2oneway}
If a pseudorandom quantum states generator with $m\ge cn$ for a constant $c>1$ exists, then a one-way quantum states generator exists.
\end{lemma}

\begin{proof}[Proof of Lemma~\ref{lemma:hiding2oneway}]
Assume that $\Pr[\mathsf{Exp}=1]$ of the security game of Definition~\ref{definition:onewayness} with $G=\StateGen$ is non-negligible. Then we can construct an adversary $\cA'$ that breaks the security of PRSGs as follows.
Let $b'\in\bit$ be the parameter of the security game for PRSGs.
\begin{enumerate}
\item
If $b'=0$, the challenger $\cC'$ of the security game for PRSGs chooses $k\leftarrow\bit^n$, runs $|\phi_k\rangle\leftarrow\StateGen(k)$ $t+1$ times, and sends $|\phi_k\rangle^{\otimes t+1}$ to $\cA'$.
If $b'=1$, the challenger $\cC'$ of the security game for PRSGs sends $t+1$ copies of Haar random state $|\psi\rangle^{\otimes t+1}$ to $\cA'$.
In other words, $\cA'$ receives $\rho^{\otimes t+1}$, where $\rho=|\phi_k\rangle$ if $b'=0$ and $\rho=|\psi\rangle$ if $b'=1$.
\item
$\cA'$ sends $\rho^{\otimes t}$ to $\cA$.
\item
$\cA$ outputs $\sigma\in\bit^n$.
\item
$\cA'$ measures $\rho$ with $\{|\phi_\sigma\rangle\langle\phi_\sigma|,I-|\phi_\sigma\rangle\langle\phi_\sigma|\}$.
If the result is $|\phi_\sigma\rangle\langle\phi_\sigma|$, $\cA'$ outputs 1.
Otherwise, $\cA'$ outputs 0.
\end{enumerate}
It is clear that 
\begin{eqnarray*}
\Pr[\cA'\to1|b'=0]=\Pr[\mathsf{Exp}=1].
\end{eqnarray*}
By assumption, $\Pr[\mathsf{Exp}=1]$ is non-negligible, and therefore
$\Pr[\cA'\to1|b'=0]$ is also non-negligible. 
On the other hand,
\begin{eqnarray*}
\Pr[\cA'\to1|b'=1]
&=&
\int d\mu(\psi)
\sum_{\sigma\in\bit^n}\Pr[\sigma\leftarrow \cA(|\psi\rangle^{\otimes t})]|\langle\phi_\sigma|\psi\rangle|^2\\
&\le&
\int d\mu(\psi) \sum_{\sigma\in\bit^n}|\langle\phi_\sigma|\psi\rangle|^2\\
&=&
\sum_{\sigma\in\bit^n}\langle\phi_\sigma|\Big[\int d\mu(\psi)|\psi\rangle\langle\psi|\Big]|\phi_\sigma\rangle\\
&=&
\sum_{\sigma\in\bit^n}\langle\phi_\sigma|\frac{I^{\otimes m}}{2^m}|\phi_\sigma\rangle\\
&\le&\frac{2^n}{2^m}\\
&\le&\frac{1}{2^{(c-1)n}}.
\end{eqnarray*}
Therefore, $\cA'$ breaks the security of PRSGs.
\end{proof}

\begin{remark}
For simplicity, Lemma~\ref{lemma:hiding2oneway} considers the case when no ancilla state exists in
PRSGs. 
It is easy to see that Lemma~\ref{lemma:hiding2oneway}
can be generalized to the case when PRSGs have ancilla states:
on input $k\in\bit^n$, a PRSG applies
$U_k$ on $|0...0\rangle$ to generate
$U_k|0...0\rangle=|\phi_k\rangle\otimes|\eta_k\rangle$, where $|\phi_k\rangle$
is the output of the PRSG and $|\eta_k\rangle$ is an ancilla state.
In that case, we modify Definition~\ref{definition:onewayness} in such a way that
the verification of the challenger is modified as follows: given $\sigma$, it generates
$U_\sigma|0...0\rangle=|\phi_\sigma\rangle\otimes|\eta_\sigma\rangle$ to obtain $|\eta_\sigma\rangle$, applies $U_\sigma^\dagger$ on $|\phi_k\rangle\otimes|\eta_\sigma\rangle$,
and measures all qubits in the computational basis. If the result is all zero, it accepts, i.e., the output of the experiment
is 1.
\end{remark}

\subsection{Definition of Digital Signatures with Quantum Public Keys}
We now formally define digital signatures with quantum public keys:

\begin{definition}[Digital signatures with quantum public keys (Syntax)]
\label{definition:signatures}
A signature scheme with quantum public keys is the set of algorithms $(\mathsf{Gen}_1,\mathsf{Gen}_2,\mathsf{Sign},\mathsf{Verify})$ such that
\begin{itemize}
\item
$\mathsf{Gen}_1(1^\lambda)$: It is a classical PPT algorithm that, on input the security parameter $1^\lambda$, outputs a classical secret key $sk$.
\item
$\mathsf{Gen}_2(sk)$: It is a QPT algorithm that, on input the secret key $sk$, outputs a 
quantum public key $pk$. 
\item
$\mathsf{Sign}(sk,m)$: It is a classical deterministic polynomial-time algorithm that, on input the secret key $sk$ and a message $m$, outputs a classical signature $\sigma$.
\item
$\mathsf{Verify}(pk,m,\sigma)$: It is a QPT algorithm that, on input a public key $pk$, the message $m$, and the signature $\sigma$, outputs $\top/\bot$.
\end{itemize}
\end{definition}

The perfect correctness is defines as follows:
\begin{definition}[Perfect correctness]
\label{definition:signature_correctness}
We say that a signature scheme satisfies perfect correctness if
\begin{eqnarray*}
\Pr[\top\leftarrow\mathsf{Verify}(pk,m,\sigma):sk\leftarrow\Gen_1(1^\lambda),pk\leftarrow \Gen_2(sk),\sigma\leftarrow\Sign(sk,m)]
=1
\end{eqnarray*}
for all messages $m$.
\end{definition}

The one-time security is defined as follows:
\begin{definition}[One-time security of digital signatures with quantum public keys]
\label{definition:one-time_security}
Let us consider the following security game, $\mathsf{Exp}$, between a challenger $\cC$ and a QPT adversary $\cA$:
\begin{enumerate}
\item
$\cC$ runs $sk\leftarrow\mathsf{Gen}_1(1^\lambda)$. 
\item
$\cA$ can query $pk\leftarrow \mathsf{Gen}_2(sk)$ $poly(\lambda)$ times.
\item
$\cA$ sends a message $m$ to $\cC$.
\item
$\cC$ runs $\sigma\leftarrow\mathsf{Sign}(sk,m)$, and sends $\sigma$ to $\cA$.
\item
$\cA$ sends $\sigma'$ and $m'$ to $\cC$.
\item
$\cC$ runs $v\leftarrow\mathsf{Verify}(pk,m',\sigma')$.
If $m'\neq m$ and $v=\top$, $\cC$ outputs 1. 
Otherwise, $\cC$ outputs 0. This $\cC$'s output is the output of the game.
\end{enumerate}
A signature scheme with quantum public keys is one-time secure if 
for any QPT adversary $\cA$ there exists a negligible function $\negl$ such that
\begin{eqnarray*}
\Pr[\mathsf{Exp}=1]\le\negl(\lambda).
\end{eqnarray*}
\end{definition}

\subsection{Construction}
Let $G$ be a OWSG.
Our construction of a one-time secure signature scheme with quantum public keys is as follows.
(For simplicity, we consider the case when the message space is $\bit$.)
\begin{itemize}
\item
$\mathsf{Gen}_1(1^n)$: Choose $k_0,k_1\leftarrow\bit^n$. Output $sk\coloneqq(sk_0,sk_1)$, where $sk_b\coloneqq k_b$ for $b\in\bit$.
\item
$\mathsf{Gen}_2(sk)$: Run $|\phi_{k_b}\rangle\leftarrow G(k_b)$ for $b\in\bit$.
Output $pk\coloneqq(pk_0,pk_1)$, where $pk_b\coloneqq|\phi_{k_b}\rangle$ for $b\in\bit$.
\item
$\mathsf{Sign}(sk,m)$: Output $\sigma\coloneqq sk_m$.
\item
$\mathsf{Verify}(pk,m,\sigma)$:
Measure $pk_m$ with $\{|\phi_\sigma\rangle\langle\phi_\sigma|,I-|\phi_\sigma\rangle\langle\phi_\sigma|\}$,
and output $\top$ if the result is $|\phi_\sigma\rangle\langle\phi_\sigma|$.
Otherwise, output $\bot$.
\end{itemize}
It is clear that this construction satisfies perfect correctness (Definition~\ref{definition:signature_correctness}).

\subsection{Security}
Let us show the security of our construction.
\begin{theorem}
\label{theorem:onetime_secure}
Our construction of a signature scheme is one-time secure. 
\end{theorem}

\begin{proof}[Proof of Theorem~\ref{theorem:onetime_secure}]
Let us consider the following security game, $\mathsf{Exp}$, between the challenger $\cC$ and a QPT adversary $\cA$:
\begin{enumerate}
\item
$\cC$ chooses $k_0,k_1\leftarrow\bit^n$.
\item
$\cA$ can query $|\phi_{k_b}\rangle\leftarrow G(k_b)$ $poly(n)$ times for $b\in\bit$.
\item
$\cA$ sends $m$ to $\cC$.
\item
$\cC$ sends $k_m$ to $\cA$.
\item
$\cA$ sends $\sigma$ to $\cC$.
\item
$\cC$ measures $|\phi_{k_{m\oplus1}}\rangle$ with $\{|\phi_\sigma\rangle\langle\phi_\sigma|,I-|\phi_\sigma\rangle\langle\phi_\sigma|\}$.
If the result is $|\phi_\sigma\rangle\langle\phi_\sigma|$, $\cC$ outputs 1.
Otherwise, $\cC$ outputs 0. This $\cC$'s output is the output of the game.
\end{enumerate}
Assume that our construction is not one-time secure, which means that $\Pr[\mathsf{Exp}=1]$ is non-negligible for
an adversary $\cA$ who queries both $\mathsf{Gen}_2(sk_0)$ and $\mathsf{Gen}_2(sk_1)$ $s=poly(n)$ times.
(Without loss of generality, we can assume that the numbers of $\cA$'s queries to $\mathsf{Gen}_2(sk_0)$ and $\mathsf{Gen}_2(sk_1)$ are the same.
An adversary who queries to $\mathsf{Gen}_2(sk_0)$ $s_0$ times and to $\mathsf{Gen}_2(sk_1)$ $s_1$ times can be simulated by another adversary
who queries to both $\mathsf{Gen}_2(sk_0)$ and $\mathsf{Gen}_2(sk_1)$ $s\coloneqq\max(s_0,s_1)$ times.)
Then, we can construct an adversary that breaks the security of OWSG $G$ as follows.
Let $\cC'$ and $\cA'$ be the challenger and the adversary of the security game of $G$, respectively.
\begin{enumerate}
\item
$\cC'$ chooses $k\leftarrow\bit^n$.
$\cC'$ runs $|\phi_k\rangle\leftarrow G(k)$ $s+1$ times.
$\cC'$ sends $|\phi_k\rangle^{\otimes s}$ to $\cA'$.
\item
$\cA'$ chooses $r\leftarrow\bit$.
$\cA'$ chooses $k'\leftarrow\bit^n$.
$\cA'$ runs $|\phi_{k'}\rangle\leftarrow G(k')$ $s$ times.
If $r=0$, $\cA'$ returns $(|\phi_k\rangle^{\otimes s},|\phi_{k'}\rangle^{\otimes s})$ to the query of $\cA$.
If $r=1$, $\cA'$ returns $(|\phi_{k'}\rangle^{\otimes s},|\phi_k\rangle^{\otimes s})$ to the query of $\cA$.
\item
$\cA$ sends $m\in\bit$ to $\cA'$.
\item
If $r=m$, $\cA'$ aborts.
If $r\neq m$, $\cA'$ sends $k'$ to $\cA$.
\item
$\cA$ sends $\sigma$ to $\cA'$.
\item
$\cA'$ sends $\sigma$ to $\cC'$.
\item
$\cC'$ measures $|\phi_k\rangle$ with $\{|\phi_\sigma\rangle\langle\phi_\sigma|,I-|\phi_\sigma\rangle\langle\phi_\sigma|\}$.
If the result is $|\phi_\sigma\rangle\langle\phi_\sigma|$, $\cC'$ outputs 1.
Otherwise, $\cC'$ outputs 0. 
\end{enumerate}
By a straightforward calculation, which is given below, 
\begin{eqnarray}
\Pr[\cC'\to1]=\frac{1}{2}\Pr[\mathsf{Exp}=1].
\label{exp1}
\end{eqnarray}
Therefore, 
if $\Pr[\mathsf{Exp}=1]$ is non-negligible, 
$\Pr[\cC'\to1]$ is also non-negligible,
which means that $\cA'$ breaks
the security of $G$.

Let us show Eq.~(\ref{exp1}).
In fact,
\begin{eqnarray*}
\Pr[\cC'\to1]&=&\frac{1}{2^{2n}}\sum_{k,k'\in\bit^n}\frac{1}{2}\Pr[1\leftarrow \cA(|\phi_k\rangle^{\otimes s},|\phi_{k'}\rangle^{\otimes s})]\Pr[\sigma\leftarrow \cA(k')]|\langle\phi_\sigma|\phi_k\rangle|^2\\
&&+\frac{1}{2^{2n}}\sum_{k,k'\in\bit^n}\frac{1}{2}\Pr[0\leftarrow \cA(|\phi_{k'}\rangle^{\otimes s},|\phi_k\rangle^{\otimes s})]\Pr[\sigma\leftarrow \cA(k')]|\langle\phi_\sigma|\phi_k\rangle|^2\\
&=&\frac{1}{2^{2n}}\sum_{k,k'\in\bit^n}\frac{1}{2}\Pr[1\leftarrow \cA(|\phi_k\rangle^{\otimes s},|\phi_{k'}\rangle^{\otimes s})]\Pr[\sigma\leftarrow \cA(k')]|\langle\phi_\sigma|\phi_k\rangle|^2\\
&&+\frac{1}{2^{2n}}\sum_{k,k'\in\bit^n}\frac{1}{2}\Pr[0\leftarrow \cA(|\phi_{k}\rangle^{\otimes s},|\phi_{k'}\rangle^{\otimes s})]\Pr[\sigma\leftarrow \cA(k)]|\langle\phi_\sigma|\phi_{k'}\rangle|^2\\
&=&\frac{1}{2}\Pr[\mathsf{Exp}=1].
\end{eqnarray*}
\end{proof}

\begin{remark}
For simplicity, we have assumed that the OWSG $G$ does not have any ancilla state.
We can extend the result to the case when $G$ has ancilla states.
In that case, the verification algorithm in our construction of digital signatures is modified as follows:
Given $\sigma$, first generate
$U_\sigma|0...0\rangle=|\phi_\sigma\rangle\otimes|\eta_\sigma\rangle$.
Then run $U_\sigma^\dagger$ on $pk_m\otimes|\eta_\sigma\rangle$,
and measures all qubits in the computational basis.
If all results are zero, output $\top$.
Otherwise, output $\bot$.
It is easy to check that a similar proof holds for the security of this generalized version.
\end{remark}

\ifnum\llncs=1
\bibliographystyle{myalpha} 
\bibliography{abbrev3,crypto,reference}
\else
\ifnum\arxiv=1
\newcommand{\etalchar}[1]{$^{#1}$}

\else
\bibliographystyle{alpha} 
\bibliography{abbrev3,crypto,reference}
\fi
\fi

\ifnum\cameraready=1
\else
\appendix

	\ifnum\submission=1
	\newpage
	 	\setcounter{page}{1}
 	{
	\noindent
 	\begin{center}
	{\Large SUPPLEMENTAL MATERIALS}
	\end{center}
 	}
	\setcounter{tocdepth}{2}

	\else
	
	\fi
\fi

\ifnum\cameraready=1
\else
\ifnum\submission=1
\newpage
\setcounter{tocdepth}{1}
\tableofcontents
\else
\fi
\fi

\appendix
\section{Making Opening Message Classical}
\label{sec:classical_opening}
In this Appendix, we show that general quantum non-interactive commitments can be
modified so that the opening message is classical.

Let us consider the following general non-interactive quantum commitments:
\begin{itemize}
    \item {\bf Commit phase:}
    The sender who wants to commit $b\in\bit$ generates a certain state $|\psi_b\rangle_{RC}$ on 
    the registers $R$ and $C$. The sender sends the register $C$ to the receiver.
    \item {\bf Reveal phase:}
The sender sends $b$ and the register $R$ to the receiver. 
 The receiver runs a certain verification algorithm on the registers $R$ and $C$.
\end{itemize}
Let us modify it as follows:
\begin{itemize}
    \item {\bf Commit phase:}
    The sender who wants to commit $b\in\bit$ chooses $x,z\leftarrow\bit^{|R|}$, and 
    generates the state 
    \begin{eqnarray*}
    [(X^xZ^z)_R\otimes I_C]|\psi_b\rangle_{RC},
    \end{eqnarray*}
    where $|R|$ is the size of the register $R$.
    The sender sends the registers $R$ and $C$ to the receiver.
    \item {\bf Reveal phase:}
The sender sends $(x,z)$ and $b$ to the receiver. 
 The receiver applies 
 $(X^xZ^z)_R\otimes I_C$ on the state, and
 runs the original verification algorithm on the registers $R$ and $C$.
\end{itemize}

\begin{theorem}
If the original commitment scheme is computationally hiding and statistically sum-binding,
then the modified commitment scheme is also
computationally hiding and statistically sum-binding.
\end{theorem}

\begin{proof}
Let us first show hiding. Hiding is clear because what the receiver has after the commit phase
in the modified scheme
is $\frac{I^{\otimes |R|}}{2^{|R|}}\otimes \mbox{Tr}_R(|\psi_b\rangle\langle\psi_b|_{RC})$, which is the same as that in the
original scheme.

Next let us show binding. Biding is also easy to understand. 
The most general action of a malicious sender in the modified scheme is as follows.
\begin{enumerate}
\item
The sender generates a state $|\Psi\rangle_{ERC}$ on the three registers $E$, $R$, and $C$.
The sender sends the registers $R$ and $C$ to the receiver.
\item
Given $b\in\bit$, the sender computes $(x,z)\in\bit^{|R|}\times\bit^{|R|}$. The sender sends $(x,z)$ and $b$ to the receiver.
\item
The receiver applies $X^xZ^z$ on the register $R$.
\item
The receiver runs the verification algorithm on the registers $R$ and $C$.
\end{enumerate}
Assume that this attack breaks sum-binding of the modified scheme.
Then we can construct an attack that breaks sum-binding of the original scheme as follows:
\begin{enumerate}
\item
The sender generates a state $|\Psi\rangle_{ERC}$ on the three registers $E$, $R$, and $C$.
The sender sends the register $C$ to the receiver.
\item
Given $b\in\bit$, the sender computes $(x,z)$ and applies $X^xZ^z$ on the register $R$.
The sender sends the register $R$, $b$, and $(x,z)$ to the receiver.
\item
The receiver runs the verification algorithm on the registers $R$ and $C$.
\end{enumerate}
It is easy to check that the two states on which the receiver applies the
verification algorithm are the same.
\end{proof}
\section{Equivalence of Binding Properties}
\label{sec:AQY_binding}
In this paper, we adopt sum-binding (\cref{definition:sumbinding}) as a definition of binding property of commitment schemes.  
On the other hand, the concurrent work by Ananth et al.~\cite{AQY21} introduces a seemingly stronger definition of binding, which we call AQY-binding, and shows that their commitment scheme satisfies it. 
The advantage of the AQY-binding is that it naturally fits into the security analysis of oblivious transfer in~\cite{C:BCKM21b}.
That is, a straightforward adaptation of the proofs in~\cite{C:BCKM21b} enables us to prove that a commitment scheme satisfying AQY-binding and computational hiding implies the existence of oblivious transfer and multi-party computation (MPC). 
Combined with their construction of an AQY-binding and computational hiding commitment scheme from PRSGs, they show that PRSGs imply oblivious transfer and MPC.

We found that it is already implicitly shown in~\cite{FUYZ20} that the sum-binding and AQY-binding are equivalent for non-interactive commitment schemes in a certain form called the \emph{generic form} as defined in~\cite{YWLQ15,Yan20,FUYZ20}.\footnote{We remark that it is also noted in \cite[Remark 6.2]{AQY21} that they are ``probably equivalent".} 
Since our commitment scheme is in the generic form, we can conclude that our commitment scheme also satisfies AQY-binding, and thus can be used for constructing oblivious transfer and MPC based on~\cite{C:BCKM21b}. We explain this in more detail below.

\smallskip
\noindent\textbf{Commitment schemes in the general form.}
We say that a commitment scheme is in the general form if it works as follows over registers $(C,R)$. 
\begin{enumerate}
    \item  
    In the commit phase, for generating a commitment to $b\in \bit$, the sender applies a unitary $Q_b$ on $\ket{0...0}_{C}\otimes\ket{0...0}_{R}$ and sends the $C$ register to the receiver.
    \item In the reveal phase, the sender sends the $R$ register along with the revealed bit $b$. 
    Then, the receiver applies $Q_b^\dagger$, measures both $C$ and $R$ in the computational basis, and accepts if the measurement outcome is $0...0$.
\end{enumerate}
See~\cite[Definition 2]{Yan20} for the more formal definition. 
Yan~\cite[Theorem 1]{Yan20} showed that for commitment schemes in the general form, the sum-binding is equivalent to the \emph{honest-binding}, which means $F(\sigma_0,\sigma_1)=\negl(\lambda)$, where $F$ is the fidelity
and $\sigma_b$ is the honestly generated commitment to $b$ for $b\in \bit$, i.e., $\sigma_b\coloneqq {\rm Tr}_R(Q_b|0...0\rangle\langle0...0|_{RC} Q_b^\dagger)$. 

\smallskip
\noindent\textbf{AQY-binding.}
Roughly speaking, the AQY-binding requires that there is an (inefficient) extractor $\mathcal{E}$ that extracts a committed message from the commitment and satisfies the following:
We define the following two experiments between a (possibly dishonest) sender and the honest receiver:
\begin{description}
\item[Real Experiment:]
In this experiment, the sender and receiver run the commit and reveal phases, and the experiment returns the sender's final state $\rho_S$ and the revealed bit $b$, which is defined to be $\bot$ if the receiver rejects.
\item[Ideal Experiment:]
In this experiment, after the sender sends a commitment, the extractor $\mathcal{E}$ extracts $b'$ from the commitment. After that, the sender reveals the commitment and the receiver verifiers it. Let $b$ be the revealed bit, which is defined to be $\bot$ if the receiver rejects.
The experiment returns $(\rho_S,b)$ if $b=b'$ and otherwise $(\rho_S,\bot)$ where $\rho_S$ is sender's final state. 
\end{description}
Then we require that for any (unbounded-time) malicious sender, outputs of the real and ideal experiments are statistically indistinguishable. 
See~\cite[Definition 6.1]{AQY21} for the formal definition. 

\smallskip
\noindent\textbf{Sum-binding and AQY-binding are equivalent.}

First, it is easy to see that AQY-binding implies sum-binding. By the AQY-binding, we can see that a malicious sender can reveal a commitment to $b\in \bit$ only if $\mathcal{E}$ extracts $b$ except for a negligible probability. Moreover, it is clear that $\Pr[\mathcal{E} \text{~extracts~} 0]+\Pr[\mathcal{E} \text{~extracts~} 1]\leq 1$ for any fixed commitment. 
Thus, the sum-binding follows.

We observe that the other direction is implicitly shown in \cite{FUYZ20} as explained below. 
As already mentioned, the sum-binding is equivalent to honest-binding. 
For simplicity, we start by considering the case of perfectly honest-binding, i.e., $F(\sigma_0,\sigma_1)=0$. 

First, as shown in \cite[Corollary 4]{FUYZ20}, there is an (inefficient) measurement $(\Pi_0,\Pi_1)$ that perfectly distinguishes $\sigma_0$ and $\sigma_1$ since we assume $F(\sigma_0,\sigma_1)=0$. Then, we can define the extractor $\mathcal{E}$ for the AQY-binding as an algorithm that just applies the measurement  $(\Pi_0,\Pi_1)$ and outputs the corresponding bit $b$.
It is shown in \cite[Lemma 6]{FUYZ20} that the final joint state over sender's and receiver's registers does not change even if we apply the measurement $(\Pi_0,\Pi_1)$ to the commitment before the reveal phase conditioned on that the receiver accepts. 
In the case of rejection, note that the revealed bit is treated as $\bot$ in the experiment for the AQY-binding.
Moreover, the measurement on the commitment register does not affect sender's final state since no information is sent from the receiver to the sender.  
By combining the above observations,  the joint distribution of the sender's final state and the revealed bit does not change even if we measure the commitment in $(\Pi_0,\Pi_1)$.
This means that the AQY-binding is satisfied.

For the non-perfectly honest-binding case, i.e., $F(\sigma_0,\sigma_1)=\negl(\lambda)$, we can rely on the perturbation technique. 
It is shown in \cite[Lemma 8]{FUYZ20} that for a non-perfectly honest-binding commitments characterized by unitaries $(Q_0,Q_1)$, there exist unitaries $(\tilde{Q}_0,\tilde{Q}_1)$ that characterize a perfectly honest-binding commitment scheme and are close to  $(Q_0,Q_1)$ in the sense that replacing $(Q_0,Q_1)$ with $(\tilde{Q}_0,\tilde{Q}_1)$ in any experiment only negligibly changes the output as long as the experiment calls  $(Q_0,Q_1)$ or $(\tilde{Q}_0,\tilde{Q}_1)$ polynomially many times. 
By using this, we can reduce the AQY-binding property of non-perfectly honest-binding commitment schemes to that of perfectly honest-binding commitment schemes with a negligible security loss.

\end{document}